\documentclass[11pt]{article}

\usepackage{hdb_macros}
\usepackage{fullpage}

\def\NotesOn{0}
\ifnum\NotesOn=1
\newcommand{\dnote}[1]{{\color{blue} \noindent (\textbf{Daniel: }#1)}}
\newcommand{\hnote}[1]{{\color{blue} \noindent (\textbf{Huck: }#1)}}
\newcommand{\nnote}[1]{{\color{blue} \noindent (\textbf{NSD: }#1)}}
\else
\newcommand{\dnote}[1]{}
\newcommand{\hnote}[1]{}
\newcommand{\nnote}[1]{}
\fi

\newcommand{\bt}{\widetilde{\vec{b}}}

\newcommand{\distance}{\mathrm{dist}}
\renewcommand{\vec}[1]{\ensuremath{\mathbf{#1}}}
\newcommand{\SZK}{\cc{SZK}}

\title{On the Lattice Distortion Problem}
\author{Huck Bennett\thanks{Department of Computer Science, Courant Institute of Mathematical Sciences, New York University. Email: \texttt{hbennett@cs.nyu.edu, noahsd@gmail.com}.} \and
  Daniel Dadush\thanks{Centrum Wiskunde \& Informatica, Amsterdam. Email: \texttt{dadush@cwi.nl}. Supported by the NWO Veni grant 639.071.510.} \and
  Noah Stephens-Davidowitz\footnotemark[1]~\thanks{This material is based upon work partially supported by the National Science Foundation under Grant No.~CCF-1320188. Any opinions, findings, and conclusions or recommendations expressed in this material are those of the authors and do not necessarily reflect the views of the National Science Foundation.}}
\begin{document}
\maketitle

\begin{abstract}
 We introduce and study the \emph{Lattice Distortion Problem} (LDP).
LDP asks how ``similar'' two lattices
are. I.e., what is the minimal distortion of a linear bijection between the two
lattices? LDP generalizes the Lattice Isomorphism Problem (the lattice analogue of Graph Isomorphism), which simply asks whether the minimal
distortion is one.

As our first contribution, we show that the distortion between any two lattices
is approximated up to a $n^{O(\log n)}$ factor by a simple function of their
successive minima. Our methods are constructive, allowing us to compute low-distortion mappings that are within a \nnote{Was: $2^{O(n (\log \log n)^2/\log n)}$} $2^{O(n \log \log n/\log n)}$ factor of optimal
in polynomial time and within a $n^{O(\log n)}$ factor of optimal in singly exponential time.
Our algorithms rely on a notion of basis reduction introduced by Seysen
(Combinatorica 1993), which we show is intimately related to lattice distortion.
Lastly, we show that LDP is NP-hard to approximate to within any
constant factor (under randomized reductions), by a reduction from the Shortest Vector Problem.
\end{abstract}

\section{Introduction}
\label{sec:intro}
An $n$-dimensional \emph{lattice} $\lat \subset \R^n$ is the set of all
integer linear combinations of linearly independent vectors $B = [\vec{b}_1,\ldots, \vec{b}_n]$ with $\vec{b}_i \in \R^n$. We write the
lattice generated by basis $B$ as $\lat(B) = \{\sum_{i=1}^n a_i \vec{b}_i : a_i \in
\Z \}$. 

Lattices are very well-studied classical mathematical objects
(e.g.,~\cite{minkowski1910geometrie,ConwaySloaneBook98}),
and over the past few decades, computational problems on lattices have found a
remarkably large number of applications in computer science. Algorithms for
lattice problems have proven to be quite useful,
and they have therefore been studied extensively
(e.g.,~\cite{lll/mathannal/lll82,Kan87,AKS01,journals/siamcomp/MicciancioV13}).
And, over the past twenty years, many strong cryptographic primitives have been
constructed with their security based on the (worst-case) hardness of various
computational lattice problems (e.g.,~\cite{Ajt96,MR07,GPV08, Gen09, Reg09,BV14}).

In this paper, we address a natural question: how ``similar'' are two lattices?
I.e., given lattices $\lat_1, \lat_2$, does there exist a linear bijective mapping $T:
\lat_1 \to \lat_2$ that does not change the distances between points by much? If
we insist that $T$ exactly preserves distances, then this is the \emph{Lattice
Isomorphism Problem} (LIP), which was studied
in~\cite{PleskenSouvignier97,journals/moc/SikiricSV09,conf/soda/HavivR14,LenstraSilverberg14}. We extend this to the \emph{Lattice Distortion Problem} (LDP),
which asks how well such a mapping $T$ can \emph{approximately} preserve
distances between points.

Given two lattices $\lat_1, \lat_2$, we define the \emph{distortion} between them as 
\[
\dist(\lat_1, \lat_2) = \min \set{\|T\| \|T^{-1}\|\, :\, T(\lat_1) = \lat_2} \text{ ,}
\]
where $\|T\| = \sup_{\|\vec{x}\|=1} \|T\vec{x}\|$ is the \emph{operator norm}. The quantity
$\kappa(T) = \|T\| \cdot \|T^{-1}\|$ is the \emph{condition number} of $T$,
which measures how much $T$ ``distorts distances'' (up to a fixed
scaling). It is easy to check that $\dist(\lat_1,\lat_2)$ bounds the ratio between most natural geometric parameters of $\lat_1$ and $\lat_2$ (up to scaling), and hence $\dist(\lat_1,\lat_2)$ is a strong
measure of ``similarity'' between lattices. In particular,
$\dist(\lat_1,\lat_2)=1$ if and only if $\lat_1,\lat_2$ are isomorphic (i.e., if and only if they are related by a scaled orthogonal transformation).

The \emph{Lattice Distortion Problem} (LDP) is then defined in the natural way as follows.  The input is two $n$-dimensional lattices $\lat_1, \lat_2$ (each represented by a basis), and the goal is to compute a bijective linear
transformation $T$ mapping $\lat_1$ to $\lat_2$ such that $\kappa(T) =
\dist(\lat_1,\lat_2)$.
In this work, we study the approximate search and
decisional versions of this problem, defined in the usual way. We refer to them as
$\gamma$-LDP and $\gamma$-GapLDP respectively, where $\gamma = \gamma(n) \geq 1$
is the approximation factor. (See Section~\ref{sec:LDPdef} for precise
definitions.)

\subsection{Our Contribution}

As our first main contribution, we show that the distortion between any two
lattices can be approximated by a natural function of geometric lattice
parameters. Indeed, our proof techniques are constructive, leading to our second main contribution: an algorithm that computes low-distortion mappings, with a trade-off between the running time and the approximation factor. Finally, we show hardness of approximating lattice distortion.

To derive useful bounds on the distortion between two lattices, it is
intuitively clear that one should study the ``different scales over which the two lattices live.'' A natural notion of this is given by the successive minima, which are defined as follows. The
$i^{th}$ successive minimum, $\lambda_i(\lat)$, of $\lat$ is the
minimum radius $r > 0$ such that $\lat$ contains $i$ linearly independent vectors
of norm at most $r$. For example, a lattice generated by a basis of orthogonal vectors of lengths $0 < a_1 \leq \dots \leq
a_n$ has successive minima $\lambda_i(\lat) = a_i$. Since low-distortion
mappings approximately preserve distances, it is intuitively clear that two
lattices can only be related by a low-distortion mapping if their successive
minima are close to each other (up to a fixed scaling). 

Concretely, for two $n$-dimensional lattices
$\lat_1,\lat_2$, we define
\begin{equation}
\label{def:m-intro}
M(\lat_1,\lat_2) = \max_{i \in [n]} \frac{\lambda_i(\lat_2)}{\lambda_i(\lat_1)}
\text{ ,}
\end{equation}
which measures how much we need to scale up $\lat_1$ so that its successive
minima are at least as large as those of $\lat_2$. For any linear map $T$ from $\lat_1$ to $\lat_2$, it is easy to see that
$\lambda_i(\lat_2) \leq \|T\| \lambda_i(\lat_1)$. Thus, by definition
$M(\lat_1,\lat_2) \leq \|T\|$. Applying the same reasoning for $T^{-1}$, we
derive the following simple lower bound on distortion.
\begin{equation}
\label{def:m-lb-intro}
 \dist(\lat_1,\lat_2)  \geq M(\lat_1,\lat_2) \cdot M(\lat_2,\lat_1)
\text{ .}
\end{equation}

We note that
this lower bound is tight when $\lat_1,\lat_2$ are each generated by bases of orthogonal vectors. But, it is a priori unclear if any comparable upper bound should hold for general lattices, since the successive minima are a very ``coarse'' characterization of the geometry of the lattice. Nevertheless, we show a corresponding upper bound.

\begin{theorem} 
\label{thm:distortion-bnd}
Let $\lat_1,\lat_2$ be $n$-dimensional lattices. Then, 
\[
M(\lat_1,\lat_2) \cdot M(\lat_2,\lat_1) 
\leq \dist(\lat_1,\lat_2) 
\leq n^{O(\log n)} \cdot M(\lat_1,\lat_2) \cdot M(\lat_2,\lat_1) 
\; .
\]
\end{theorem}

In particular, Theorem~\ref{thm:distortion-bnd}, together with standard transference theorems (e.g.,~\cite{bana/transference}), implies that $n^{O(\log n)}$-GapLDP is in $\NP \cap \coNP$. While the factor on the right-hand side of the theorem might be far from optimal, we show in Section~\ref{sec:densepacking} that it cannot be improved below
$\Omega(\sqrt{n})$.  Intuitively, this is because there exist lattices that are much more dense than $\Z^n$ over large scales but still have $\lambda_i(\lat) = \Theta(1)$ for all $i$. I.e., there exist very dense lattice sphere packings (see, e.g., ~\cite{siegal45}).

To prove the above theorem, we make use of the intuition that a low-distortion
mapping $T$ from $\lat_1$ to $\lat_2$ should map a ``short'' basis $B_1$ of
$\lat_1$ to a ``short'' basis $B_2$ of $\lat_2$. (Note that the
condition $TB_1 = B_2$ completely determines $T =
B_2B_1^{-1}$.) The difficulty here is that standard notions of ``short'' fail
for the purpose of capturing low-distortion mappings.  In particular, in
Section~\ref{sec:hkzlb}, we show that
Hermite-Korkine-Zolotarev (HKZ)
reduced bases, one of the strongest notions of ``shortest possible'' lattice
bases, do
not suffice by themselves for building low-distortion mappings. (See Section~\ref{sec:bases} for the definition of HKZ-reduced bases.) In particular,
we give a simple example of a lattice $\lat$ where an HKZ-reduced basis of $\lat$ misses the optimal distortion
$\dist(\Z^n,\lat)$ by an exponential factor.

Fortunately, we show that a suitable notion of shortness does exist for
building low-distortion mappings by making
 a novel connection between low-distortion mappings and a notion of basis reduction introduced by Seysen~\cite{journals/combinatorica/Seysen93}. In particular, for a basis $B = [\vec{b}_1,\ldots, \vec{b}_n]$ and dual basis $B^* = B^{-T} =
[\vec{b}_1^*,\dots,\vec{b}_n^*]$, Seysen's condition number is defined as
\[
S(B) = \max_{i \in [n]} \|\vec{b}_i\| \|\vec{b}_i^*\| \text{ .}
\]
Note that we always have $\iprod{\vec{b}_i, \vec{b}_i^*} = 1$, so this parameter measures how tight the Cauchy-Schwarz inequality is over all primal-dual basis-vector pairs. We extend this notion and define $S(\lat)$ as the minimum of $S(B)$ over all bases $B$ of $\lat$. Using this notion, we give an effective version of
Theorem~\ref{thm:distortion-bnd} as follows.

\begin{theorem}
\label{thm:basis-distortion-bnd}
Let $\lat_1,\lat_2$ be $n$-dimensional lattices. Let $B_1,B_2 \in \R^{n \times
n}$ be bases of $\lat_1,\lat_2$ whose columns are sorted in non-decreasing order
of length. Then, we
have that
\[
M(\lat_1,\lat_2)M(\lat_2,\lat_1) \leq \kappa(B_2 B_1^{-1}) \leq n^2 S(B_1)^2 S(B_2)^2
\cdot M(\lat_1,\lat_2)M(\lat_2,\lat_1) \text{ .}
\]
In particular, we have that
\[
M(\lat_1,\lat_2)M(\lat_2,\lat_1) \leq \dist(\lat_1,\lat_2) \leq n^2 S(\lat_1)^2
S(\lat_2)^2 \cdot M(\lat_1,\lat_2)M(\lat_2,\lat_1) \text{ .}
\]  
\end{theorem}  

From here, the bound in Theorem~\ref{thm:distortion-bnd} follows directly from
the following (surprising) theorem of Seysen.

\begin{theorem}[Seysen~\cite{journals/combinatorica/Seysen93}] 
\label{thm:seysenbound}
For any $\lat \subset \R^n$, $S(\lat) \leq n^{O(\log n)}$.
\end{theorem}

This immediately yields an algorithm for approximating the distortion between two lattices, by using standard lattice algorithms to approximate $M(\lat_1,\lat_2)$ and $M(\lat_2,\lat_1)$. 
But, Seysen's proof of the above theorem is actually constructive! In particular, he shows how
to efficiently convert any suitably reduced lattice basis
into a basis with a low Seysen condition number. (See Section~\ref{sec:seysen} for details.) Using this methodology, combined with
standard basis reduction techniques, we derive the following
time-approximation trade-off for $\gamma$-LDP.

\begin{theorem}[Algorithm for LDP]
\nnote{Was: For any $\log n \leq k \leq n$, there is an algorithm solving $k^{O(n \log k/k)}$-LDP in time $2^{O(k)} $. Similarly, there is an algorithm solving $(n^{O(\log n)}k^{O(n/k)})$-GapLDP in time $2^{O(k)}$.}
For any $\log n \leq k \leq n$, there is an algorithm solving $k^{O(n/k + \log n)}$-LDP in time $2^{O(k)} $.
\label{thm:approx-informal}
\end{theorem}

In other words, using the bounds in Theorem~\ref{thm:distortion-bnd} together with known algorithms, we are able to approximate the distortion between two lattices. But, with a bit more work, we are able to solve \emph{search} LDP by explicitly computing a low-distortion mapping between the input lattices. \nnote{Snip: Interestingly, we achieve a (slightly) better approximation factor for the decision variant than the search variant. E.g., plugging in $k = O(\log n)$ gives an efficient algorithm that achieves an approximation factor of $2^{O(n \log \log n/ \log n)}$ for GapLDP, but only an approximation factor of $2^{O(n (\log \log n)^2/ \log n)}$ for \emph{search} LDP. 
This small gap in the approximation factors arises from some rather technical aspects of basis reduction.\footnote{In particular, ``HKZ-like basis reduction'' tries to iteratively make each Gram-Schmidt vector as short as possible. ``LLL-like basis reduction'' bounds how quickly the lengths of the Gram-Schmidt vectors can decay. We can approximate most lattice problems to within a factor of $\gamma = \poly(n) \cdot k^{O(n/k)}$ in time $2^{O(k)}$ by using a hybrid of these two approaches to build $\gamma$-approximate HKZ bases~\cite{Schnorr87}. But, the resulting decay of the length of the Gram-Schmidt vectors is slightly larger than $\gamma$; it is $k^{O(n \log k/k)}$. See Section~\ref{sec:bases} for the details.}}

We also prove the following lower bound for LDP.

\begin{theorem}[Hardness of LDP]
$\gamma$-GapLDP is NP-hard under randomized polynomial-time reductions for any
constant $\gamma \geq 1$.
\label{thm:hardness-informal}
\end{theorem}

In particular, we show a reduction from approximating the (decisional)
Shortest Vector Problem (GapSVP) over lattices to $\gamma$-GapLDP, where the
approximation factor that we obtain for GapSVP is $O(\gamma)$. Since hardness of GapSVP is quite well-studied~\cite{Ajtai-SVP-hard,Mic01svp,Khot05svp,HRsvp}, we are immediately able to import many hardness results to GapLDP. (See Corollary~\ref{cor:SVPtoLDP} and
Theorem~\ref{thm:LDPhard} for the precise statements.)

\subsection{Comparison to related work}

The main related work of which we are aware is that of Haviv and
Regev~\cite{conf/soda/HavivR14} on the Lattice Isomorphism Problem (LIP). In
their paper, they give an $n^{O(n)}$-time algorithm for solving LIP exactly,
which proceeds by cleverly identifying a small candidate set of bases of
$\lat_1$ and $\lat_2$ that must be mapped to each
other by any isomorphism. One might expect that such
an approach should also work for the purpose of solving LDP either exactly or
for approximation factors below $n^{O(\log n)}$. However, the crucial
assumption in LIP, that vectors in one lattice must be mapped to vectors of
the same length in the other, completely breaks down in the current context. We thus do not
know how to extend their techniques to LDP. 

Much more generally, we note that LIP is closely related to the Graph Isomorphism Problem (GI). For example, both problems are in $\SZK$ but not known to be in $\P$ (although recent work on algorithms for GI has been quite exciting~\cite{BabaiGI16}!), and GI reduces to LIP~\cite{journals/moc/SikiricSV09}. Therefore, LDP is qualitatively similar to the Approximate Graph Isomorphism Problem, which was studied by Arora, Frieze, and Kaplan~\cite{SFK02}, who showed an upper bound, and Arvind, K{\"o}bler, Kuhnert, and Vasudev~\cite{AKKV12}, who proved both upper and lower bounds. In particular,~\cite{AKKV12} showed that various versions of this problem are $\NP$-hard to approximate to within a constant factor. Qualitatively, these hardness results are similar to our Theorem~\ref{thm:hardness-informal}.

\subsection{Conclusions and Open Problems}
In conclusion, we introduce the Lattice Distortion Problem and show a connection between LDP
and the notion of Seysen-reduced bases. We use this connection to derive
time-approximation trade-offs for LDP. We also prove approximation hardness for
GapLDP, showing a qualitative difference with LIP (which is unlikely to be $\NP$-hard under reasonable complexity theoretic assumptions).

One major open question is what the correct bound in Theorem~\ref{thm:seysenbound} is. In particular, there are no known families of
lattices for which the Seysen condition number is provably superpolynomial, and
hence it is possible that $S(\lat) = \poly(n)$ for any $n$-dimensional lattice
$\lat$.
A better bound would
immediately improve our Theorem~\ref{thm:basis-distortion-bnd} and give a better approximation factor for GapLDP.

We also note that all of our algorithms solve LDP
for arguably very large approximation factors $n^{\Omega(\log n)}$. We currently do not even know
whether there exists a fixed-dimension polynomial-time algorithm for
$\gamma$-LDP for any $\gamma = n^{o(\log n)}$. The main problem here is that we
do not have any good characterization of nearly optimal
distortion mappings between lattices.  

\subsection*{Organization}

In Section~\ref{sec:prelims}, we present necessary background material. In
Section~\ref{sec:approx}, we give our approximations for lattice distortion,
proving Theorems~\ref{thm:basis-distortion-bnd} and~\ref{thm:approx-informal}.
In Section~\ref{sec:hardness}, we give the hardness for lattice distortion,
proving Theorem~\ref{thm:hardness-informal}. In Section~\ref{sec:examples}, we give some illustrative example instances of lattice distortion.

\subsection*{Acknowledgements}
We thank Oded Regev for pointing us to Seysen's paper and for many helpful conversations. The concise proof of Lemma~\ref{lem:sb_succ_dual} using the Transference Theorem is due to Michael Walter. We thank Paul Kirchner for telling us about Proposition~\ref{prop:slide_eta}, and for identifying a minor bug in an earlier version of this paper.

\section{Preliminaries}
\label{sec:prelims}

For $\vec{x} \in \R^n$, we write $\|\vec{x}\|$ for the Euclidean norm of $\vec{x}$. We omit any mention of the bit length in the running time of our algorithms. In particular, all of our algorithms take as input vectors in $\Q^n$ and run in time $f(n) \cdot \poly(m)$ for some $f$, where $m$ is the maximal bit length of an input vector. We therefore suppress the factor of $\poly(m)$.

\subsection{Lattices}
The $i^{th}$ \emph{successive minimum} of a lattice $\lat$ is defined as $\lambda_i(\lat) = \inf \set{r > 0 : \dim(\lspan(r B_2^n \cap \lat)) \geq i}$. That is, the first successive minimum is the length of the shortest non-zero lattice vector, the second successive minimum is the length of the shortest lattice vector which is linearly independent of a vector achieving the first, and so on. When $\lat$ is clear from context, we simply write $\lambda_i$. %

The \emph{dual lattice} of $\lat$ is defined as $\lat^* = \set{\vec{x} \in \R^n : \forall \vec{y} \in \lat \; \langle \vec{x}, \vec{y} \rangle \in \Z}$.
If $\lat = \lat(B)$ then $\lat^* = \lat(B^*)$ where $B^* = B^{-T}$, the inverse transpose of $B$.
We call $B^* = [\vec{b}_1^*, \ldots, \vec{b}_n^*]$ the \emph{dual basis} of $B$, and write $\lambda_i^* = \lambda_i(\lat^*)$. We will repeatedly use Banaszczyk's Transference Theorem, which relates the successive minima of a lattice to those of its dual.

\begin{theorem}[Banaszczyk's Transference Theorem \cite{bana/transference}]
For every rank $n$ lattice $\lat$ and every $i \in [n]$, $1 \leq \lambda_i(\lat) \lambda_{n-i+1}(\lat^*) \leq n$.
\label{thm:transference}
\end{theorem}

Given a lattice $\lat$, we define the \emph{determinant} of $\lat$ as $\det(\lat) := |\det(B)|$, where $B$ is a basis with $\lat(B) = \lat$. Since two bases $B,B'$ of $\lat$ differ by a unimodular transformation, we have that $|\det(B)| = |\det(B')|$ so that $\det(\lat)$ is well-defined.

\nnote{Added: } We sometimes work with lattices that do not have full rank---i.e., lattices generated by $d$ linearly independent vectors $\lat = \lat(\vec{b}_1,\ldots, \vec{b}_d)$ with $d < n$. In this case, we simply identify $\mathrm{span}(\vec{b}_1,\ldots, \vec{b}_d)$ with $\R^d$ and consider the lattice to be embedded in this space.

\subsection{Linear mappings between lattices}

We next characterize linear mappings between lattices in terms of bases.

\begin{lemma}
  Let $\lat_1, \lat_2$ be full-rank lattices. Then a mapping $T: \lat_1 \to \lat_2$ is bijective and linear if and only if $T = BA^{-1}$ for some bases $A, B$ of $\lat_1, \lat_2$ respectively. In particular, for any basis $A$ of $\lat_1$, $T(A)$ is a basis of $\lat_2$.
\end{lemma}
\dnote{Put proof in appendix}
\begin{proof}
  We first show that such a mapping is a bijection from $\lat_1$ to $\lat_2$. Let $T = BA^{-1}$ where $A = [\vec{a}_1,\ldots, \vec{a}_n]$ and  $B = [\vec{b}_1,\ldots, \vec{b}_n]$ are bases of $\lat_1, \lat_2$ respectively. Because $T$ has full rank, it is injective as a mapping from $\R^n$ to $\R^n$, and it is therefore injective as a mapping from $\lat_1$ to $\lat_2$. We have that for every $\vec{w} \in \lat_2$, $\vec{w} = \sum_{i=1}^n c_i \vec{b}_i$ with $c_i \in \Z$. Let $\vec{v} = \sum_{i=1}^n c_i \vec{a}_i \in \lat_1$. Then, $T(\vec{v}) = T (\sum_{i=1}^n c_i \vec{a}_i) = \sum_{i=1}^n c_i \vec{b}_i = \vec{w}$. Therefore, $T$ is a bijection from $\lat_1$ to $\lat_2$.

  We next show that any linear map $T$ with $T(\lat_1) = \lat_2$ must have this form. Let $A = [\vec{a}_1, \ldots, \vec{a}_n]$ be a basis of $\lat_1$, and let $B = T(A)$. We claim that $B = [\vec{b}_1, \ldots, \vec{b}_n]$ is a basis of $\lat_2$.

  Let $\vec{w} \in \lat_2$. Because $T$ is a bijection between $\lat_1$ and $\lat_2$, there exists $\vec{v} \in \lat_1$ such that $T\vec{v} = w$. Using the definition of a basis and the linearity of $T$,
\[
\vec{w} = T\vec{v} = T \big( \sum_{i=1}^n c_i \vec{a}_i \big) = \sum_{i=1}^n c_i \vec{b}_i,
\]
for some $c_1, \ldots, c_n \in \Z$. Because $\vec{w}$ was picked arbitrarily, it follows that $B$ is a basis of $\lat_2$.
  
\end{proof}

\subsection{Seysen's condition number \texorpdfstring{$S(B)$}{S(B)}}
\label{sec:seysendef}

Seysen shows how to take any basis with relatively low multiplicative drop in its Gram-Schmidt vectors and convert it into a basis with relatively low $S(B) = \max_i \|\vec{b}_i\|\|\vec{b}_i^*\|$~\cite{journals/combinatorica/Seysen93}. By combining this with \nnote{Was: standard basis reduction algorithms~\cite{Schnorr87,AKS01}} Gama and Nguyen's slide reduction technique~\cite{GN08}, we obtain the following result.

\dnote{It make sense to include the definition of $\eta(B)$ here and the
$n^{O(\log n)}$ bound here? I.e. the statement of Theorem~\ref{thm:seysen-main},
using $n^{O(\log n)}$ instead of $\zeta(n)^2$? I refer to this in the intro, so
it might be good to explain quickly here.}

\begin{theorem}
  For every $\log n \leq k \leq n$ there exists an algorithm that takes a lattice $\lat$ as input and computes a basis $B$ of $\lat$ with \nnote{Was: $S(B) \leq k^{O(n \log k/k)}$} $S(B) \leq k^{O(n/k + \log k)}$ in time $2^{O(k)}$. \nnote{Note that $k^{O(n/k + \log k)}$ is the same thing as $n^{O(\log n)} \cdot k^{O(n/k)}$. I thought this way of writing it was slightly nicer.}
\label{thm:seysen-slide}
\end{theorem}
In particular, applying Seysen's procedure to slide-reduced bases suffices. We include a proof of Theorem~\ref{thm:seysen-slide} and a high-level description of Seysen's procedure in Section~\ref{sec:bases}.

\subsection{The Lattice Distortion Problem}
\label{sec:LDPdef}

\begin{definition}
  For any $\gamma = \gamma(n) \geq 1$, the $\gamma$-Lattice Distortion Problem ($\gamma$-LDP) is the search problem defined as follows. The input consists of two lattices $\lat_1, \lat_2$ (represented by bases $B_1, B_2 \in \Q^{n \times n}$). The goal is to output a matrix $T \in \R^{n \times n}$ such that $T(\lat_1) = \lat_2$ and $\kappa(T) \leq \gamma \cdot \dist(\lat_1, \lat_2)$.
\label{def:ldp}
\end{definition}

\begin{definition}
  For any $\gamma = \gamma(n) \geq 1$, the $\gamma$-GapLDP is the promise problem defined as follows. The input consists of two lattices $\lat_1, \lat_2$ (represented by bases $B_1, B_2 \in \Q^{n \times n}$) and a number $c \geq 1$. The goal is to decide between a `YES' instance where $\dist(\lat_1, \lat_2) \leq c$ and a `NO' instance where $\dist(\lat_1, \lat_2) > \gamma \cdot c$.
\label{def:gapldp}
\end{definition}

\subsection{Complexity of LDP}

We show some basic facts about the complexity of GapLDP. First, we show that the Lattice Isomorphism Problem (LIP) corresponds to the special case of GapLDP where $c=1$. LIP takes bases of $\lat_1, \lat_2$ as input and asks if there exists an orthogonal linear transformation $O$ such that $O(\lat_1) = \lat_2$. Haviv and Regev~\cite{conf/soda/HavivR14} show that there exists an $n^{O(n)}$-time algorithm for LIP, and that LIP is in the complexity class $\SZK$.

\dnote{Proofs go to the appendix}
\begin{lemma}
There is a polynomial-time reduction from LIP to $1$-GapLDP.
\end{lemma}

\begin{proof}
Let $\lat_1, \lat_2$ be an LIP instance. First check that $\det(\lat_1) = \det(\lat_2)$. If not, then output a trivial `NO' instance of $1$-GapLDP. Otherwise, map the LIP instance to the $1$-GapLDP instance with the same input bases and $c = 1$. For any $T : \lat_1 \to \lat_2$, we must have $\det(T) = 1$, and therefore $\kappa(T) = 1$ if and only if $\norm{T} = \norm{T^{-1}} = 1$. So, this is a `YES' instance of GapLDP if and only if $\lat_1, \lat_2$ are isomorphic.
\end{proof}

\begin{lemma}
  $1$-GapLDP is in $\NP$.
  \label{lem:ldp-np}
\end{lemma}

\begin{proof}
  Let $I = (\lat_1, \lat_2, c)$ be an instance of GapLDP, and let $s$ be the length of $I$. We will show that for a `YES' instance, there are bases $A, B$ of $\lat_1, \lat_2$ respectively such that $T = BA^{-1}$ requires at most $\poly(s)$ bits to specify and $\kappa(T) \leq c$. Assume without loss of generality that $\lat_1, \lat_2 \subseteq \Z^n$. Otherwise, scale the input lattices to achieve this at the expense of a factor $s$ blow-up in input size.

  To satisfy $\norm{T}\norm{T^{-1}} \leq c$, we must have that $\abs{t_{ij}} \leq \norm{T} \leq c \cdot \det(\lat_2)/\det(\lat_1) \leq c \cdot \det(\lat_2)$ for each entry $t_{ij}$ of $T$. By Cramer's rule, each entry of $A^{-1}$ and hence $T$ will be an integer multiple of $\frac{1}{\det{\lat_1}}$, so we can assume without loss of generality that the denominator of each entry of $T$ is $\det{\lat_1}$.
  
  Combining these bounds and applying Hadamard's inequality, we get that $\abs{t_{ij}}$ takes at most 
  \[
  \log(c \cdot \det(\lat_1) \det(\lat_2)) \leq \log\Big(c \cdot \prod_{i=1}^n \norm{\vec{a}_i} \prod_{i=1}^n \norm{\vec{b}_i}\Big)
  \] 
  bits to specify. Accounting for the sign of each $t_{ij}$, it follows that $T$ takes at most $n^2 \cdot \log(2c \cdot \prod_{i=1}^n \norm{\vec{a}_i} \norm{\vec{b}_i}) \leq n^2 \cdot (s + 1)$ bits to specify.
\end{proof}

We remark that we can replace $c$ with the quantity $n^{O(\log n)} M(\lat_1,\lat_2) M(\lat_2,\lat_1)$ (as given by the upper bound in Theorem~\ref{thm:distortion-bnd}) in the preceding argument to obtain an upper bound on the distortion of an \emph{optimal} mapping $T$ that does not depend on $c$.

\subsection{Basis reduction}
\label{sec:bases}
In this section, we define various notions of basis reductions and show how to use them to prove Theorem~\ref{thm:seysen-slide}.

For a basis $B = [\vec{b}_1,\ldots, \vec{b}_n]$, we write $\pi_i^{(B)} := \pi_{\{\vec{b}_1, \ldots, \vec{b}_{i-1}\}^\perp}$ to represent projection onto the subspace $\{\vec{b}_1, \ldots, \vec{b}_{i-1}\}^\perp$. We then define the \emph{Gram-Schmidt} orthogonalization of $B$, $(\bt_1,\ldots, \bt_n)$ as $\bt_i = \pi_{i}^{(B)}(\vec{b}_i)$. By construction the vectors $\bt_1, \ldots, \bt_n$ are orthogonal, and each $\vec{b}_i$ is a linear combination of $\bt_1, \ldots, \bt_i$. We define
$\mu_{ij} = \frac{\langle \vec{b}_i, \bt_j \rangle}{\langle \bt_j, \bt_j \rangle}$.

We define the QR-decomposition of a full-rank matrix $B$ as $B = QR$ where $Q$ has orthonormal columns, and $R$ is upper triangular. The QR-decomposition of a matrix is unique, and can be computed efficiently by applying Gram-Schmidt orthogonalization to the columns of $B$.

Unimodular matrices, denoted $GL(n, \Z)$, form the multiplicative group of $n \times n$ matrices with integer entries and determinant $\pm 1$.

\begin{fact}
$\lat(B) = \lat(B')$ if and only if there exists $U \in GL(n, \Z)$ such that $B' = B \cdot U$.
\label{fct:unimodular}
\end{fact}

Based on this, a useful way to view basis reduction is as right-multiplication by unimodular matrices.

\subsubsection{\nnote{Replaced BKZ with this:} Slide-reduced bases}
\label{sec:slide}

A very strong notion of basis reduction introduced by Korkine and Zolotareff~\cite{kzbases} gives one way of formalizing what it means to be a ``shortest-possible'' lattice basis. 

\begin{definition}[\cite{kzbases}, Definition 1 in~\cite{journals/combinatorica/Seysen93}]
  Let $B$ be a basis of $\lat$. $B = [\vec{b}_1,\ldots, \vec{b}_n]$ is HKZ (Hermite-Korkine-Zolotareff) reduced if

  \begin{enumerate}
  \item $\forall j < i,\;  \abs{\mu_{ij}} \leq \half$;
  \item $\norm{b_1} = \lambda_1(\lat(B))$; and
  \item if $n > 1$, then $[\pi_2^{(B)}(\vec{b}_2), \ldots, \pi_2^{(B)}(\vec{b}_n)]$ is an HKZ basis of $\pi_2^{(B)}(\lat)$.
  \end{enumerate}
  \label{def:hkzbases}
\end{definition}

By definition, the first vector $\vec{b}_1$ in an HKZ basis is a shortest vector in the lattice. Furthermore, computing an HKZ basis can be achieved by making $n$ calls to an SVP oracle. So, the two problems have the same time complexity up to a factor of $n$. In particular, computing HKZ bases is $\NP$-hard.

Gama and Nguyen (building on the work of Schnorr~\cite{Schnorr87}) introduced the notion of slide-reduced bases \cite{GN08}, which can be thought of as a relaxed notion of HKZ bases that can be computed more efficiently. 

\begin{definition}[{\cite[Definition 1]{GN08}}]
\label{def:slide}
  Let $B$ be a basis of $\lat \subset \Q^n$ and $\eps > 0$. We say that $B$ is $\eps$-DSVP (dual SVP) reduced if its corresponding dual basis $[\vec{b}_1^*,\ldots, \vec{b}_n^*]$ satisfies $\|\vec{b}_n^*\| \leq (1+\eps) \cdot \lambda_1(\lat^*)$.
  
  Then, for $k\geq 2$ an integer dividing $n$, we say that $B = [\vec{b}_1,\ldots, \vec{b}_n]$ is $(\eps, k)$-slide reduced if
\begin{enumerate}
\item $\forall j < i,\;  \abs{\mu_{ij}} \leq \half$;
\item $\forall 0 \leq i \leq n/k-1$, the ``projected truncated basis'' $[\pi_{ik+1}^{(B)}(\vec{b}_{ik+1}), \ldots, \pi_{ik+1}^{(B)}(\vec{b}_{ik+k})]$ is HKZ reduced; and
\item \label{item:DSVP} $\forall 0 \leq i \leq n/k-2$, the ``shifted projected truncated basis'' $[\pi_{ik+2}^{(B)}(\vec{b}_{ik+2}), \ldots, \pi_{ik+2}^{(B)}(\vec{b}_{ik+k+1})]$ is $\eps$-DSVP reduced.
\end{enumerate}
\end{definition}

\begin{theorem}[{\cite{GN08}}]
\label{thm:slide}
There is an algorithm that takes as input a lattice $\lat \subset \Q^n$, $\eps > 0$, and integer $k \geq \log n$ dividing $n$ and outputs a $(k,\eps)$-slide-reduced basis of $\lat$ in time $\poly(1/\eps) \cdot 2^{O(k)}$.
\end{theorem}

We will be particularly concerned with the ratios between the length of the Gram-Schmidt vectors of a given basis. We prefer bases whose Gram-Schmidt vectors do not ``decay too quickly,'' and we measure this decay by
\[
  \eta(B) = \max_{i \leq j} \frac{\norm{\bt_i}}{\norm{\bt_j}}
  \; .
\]
Previous work bounded $\eta(B)$ for HKZ-reduced bases as follows.

\begin{theorem}[{\cite[Proposition 4.2]{journals/combinatorica/LagariasLS90}}]
For any HKZ-reduced basis $B$ over $\Q^n$, $\eta(B) \leq n^{O(\log n)}$.
\label{thm:hkz-gs}
\end{theorem}

We now use Theorem~\ref{thm:hkz-gs} and some of the results in \cite{GN08} to bound $\eta(B)$ of slide-reduced bases.

\begin{proposition}
\label{prop:slide_eta}
For any integer $k \geq 2$ dividing $n$, if $B$ is an $(1/n,k)$-slide-reduced basis for a lattice $\lat \subset \Q^n$, then $\eta(B) \leq k^{O(n/k + \log k)}$.
\end{proposition}
\begin{proof}
We collect three simple inequalities that will together imply the result. First, from \cite[Eq.~(16)]{GN08}, we have $\|\bt_1\| \leq k^{O(n/k)} \cdot \|\bt_{jk + 1}\|$ for all $0 \leq j \leq n/k-1$. Noting that the projection $[\pi_{ik+1}(\vec{b}_{ik+1}),\ldots, \pi_{ik+1}(\vec{b}_{ik+k})]$ of a slide-reduced basis is also slide reduced, we see that 
\begin{equation}
\label{eq:between_block_decay}
\|\bt_{ik+1}\| \leq k^{O(n/k)} \cdot \|\bt_{jk + 1}\| \;,
\end{equation}
for all $0 \leq i \leq j \leq n/k-1$. Next, from Theorem~\ref{thm:hkz-gs} and the fact that the ``projected truncated bases'' are HKZ reduced, we have that 
\begin{equation}
\label{eq:HKZ_decay}
\|\bt_{ik+\ell}\| \leq k^{O(\log k)} \cdot \|\bt_{ik+\ell'}\|
\; ,
\end{equation}
 for all $1 \leq \ell \leq \ell' \leq k$. Finally, \cite{GN08} observe that\footnote{They actually observe that a slide-reduced basis is LLL reduced, which immediately implies Eq.~\eqref{eq:last_vector_decay}.}
\begin{equation}
\label{eq:last_vector_decay}
\| \bt_{ik+k} \| \leq C \cdot \|\bt_{ik+k+1}\|
\; ,
\end{equation}
for all $0 \leq i \leq n/k-2$, where $C > 0$ is some universal constant.

Now, let $0 \leq i \leq i' \leq n/k-1$ and $1 \leq \ell, \ell' \leq k$ such that $ik + \ell < i' k + \ell'$. If $i = i'$, then clearly $\|\bt_{ik + \ell}\|/\|\bt_{i'k + \ell'}\| \leq k^{O(\log k)}$ by Eq.~\eqref{eq:HKZ_decay}. Otherwise, $i < i'$ and
\begin{align*}
\frac{\|\bt_{ik + \ell}\|}{\|\bt_{i'k + \ell'}\|} 
&\leq k^{O(\log k)} \cdot \frac{\|\bt_{ik + \ell}\|}{\|\bt_{i'k + 1}\|} &\text{(Eq.~\eqref{eq:HKZ_decay})}\;\\
&\leq k^{O(n/k + \log k)} \cdot \frac{\|\bt_{ik + \ell}\|}{\|\bt_{ik +k+ 1}\|} &\text{(Eq.~\eqref{eq:between_block_decay})}\;\\
&\leq k^{O(n/k + \log k)} \cdot \frac{\|\bt_{ik + \ell}\|}{\|\bt_{ik +k}\|} &\text{(Eq.~\eqref{eq:last_vector_decay})}\;\\
&\leq k^{O(n/k + \log k)} &\text{(Eq.~\eqref{eq:HKZ_decay})}
,
\end{align*}
as needed.
\end{proof}

Finally, we show how to get rid of the requirement that $k$ divides $n$.

\begin{proposition}
\label{prop:slide_alg}
For any $\log n \leq k \leq n$, there is an algorithm that takes as input a lattice $\lat \subset \Q^n$ and outputs a basis $B$ of $\lat$ such that 
$
\eta(B) \leq k^{O(n/k + \log k)}$. Furthermore, the algorithm runs in time $2^{O(k)}$.
\end{proposition}
\begin{proof}
We assume without loss of generality that $k$ is an integer. Let $n' = \ceil{n/k} \cdot k$ be the smallest integer greater than or equal to $n$ that is divisible by $k$. On input a basis $\hat{B} = [\hat{\vec{b}}_1,\ldots,\hat{\vec{b}}_n]$ for the lattice $\lat \subset \Q^n$, the algorithm behaves as follows. Let $r = 2^{\Omega(n^2)} \cdot \max_i \|\hat{\vec{b}}_i\|$. Let $\lat' := \lat(\hat{\vec{b}}_1,\ldots, \hat{\vec{b}}_n, r \cdot \vec{e}_{n+1}, \ldots, r \cdot \vec{e}_{n'}) \subset \Q^{n'}$ be ``the lattice obtained by appending $n'-n$ orthogonal vectors of length $r$ to $\lat$.'' The algorithm then computes a basis $B' = [\vec{b}_1, \ldots, \vec{b}_{n'}]$ of $\lat'$ as in Theorem~\ref{thm:slide} with $\eps = 1/n$ and returns the basis consisting of the first $n$ entries of $B'$, $B = [\vec{b}_1,\ldots, \vec{b}_n]$.

It follows immediately from Theorem~\ref{thm:slide} that the running time is as claimed, and from Proposition~\ref{prop:slide_eta} we have that $\eta(B) \leq \eta(B') \leq k^{O(n'/k + \log k)} \leq k^{O(n/k + \log k)}$. So, we only need to prove that $B$ is in fact a basis for $\lat$ (as opposed to some other sublattice of $\lat'$). 

Consider the first $i$ such that $\|\bt_i\| \geq r$. \nnote{Rewrote: } We claim that $\pi_{i}^{(B')}(\hat{B}) = 0$. If not, then choose $j$ such that $\pi_{i}^{(B')}(\hat{\vec{b}}_j) \neq \vec0$.  There must be some $ak + \ell \geq i$ with $1 \leq \ell \leq k$ such that $\pi_{ak+\ell}^{(B')}(\hat{\vec{b}}_j) \neq \vec0$, but $\pi_{ak+\ell}^{(B')}(\hat{\vec{b}}_j) \in \mathrm{span}(\pi_{ak+\ell}(\vec{b}_{ak+\ell}), \ldots, \pi_{ak+\ell}(\vec{b}_{ak+k}))$. It follows from the fact that $B'$ is a basis of $\lat'$ that $\pi_{ak+\ell}^{(B')}(\hat{\vec{b}}_j) \in \lat(\pi_{ak+\ell}(\vec{b}_{ak+\ell}), \ldots, \pi_{ak+\ell}(\vec{b}_{ak+k}))$. But, since $[\pi_{ak+\ell}(\vec{b}_{ak+\ell}), \ldots, \pi_{ak+\ell}(\vec{b}_{ak+k})]$ is an HKZ basis, it must be the case that 
\[
\|\bt_{ak+\ell}\| = \|\pi_{ak+\ell}(\vec{b}_{ak+\ell})\| \leq \|\pi_{ak+\ell}^{(B')}(\hat{\vec{b}}_j)\| \leq \|\hat{\vec{b}}_j\| \leq \frac{r}{2^{\Omega(n^2)}} < \frac{\|\bt_i\|}{\eta(B')}
\; ,
\]
which contradicts the definition of $\eta(B')$.

So, $\pi_{i}^{(B')}(\hat{B}) = 0$. It follows that $i = n+1$ and $\lat \subset \mathrm{span}(\vec{b}_1,\ldots, \vec{b}_{i-1}) = \mathrm{span}(B)$. And, since $B'$ is a basis of $\lat'$, it follows that $\lat = \lat(B)$, as needed.

\end{proof}

\subsubsection{Seysen bases}
\label{sec:seysen}
Although \nnote{Was: BKZ-reduced} slide-reduced bases $B$ consist of short vectors and have bounded $\eta(B)$, they make \nnote{Was: no guarantees} only weak guarantees about the length of vectors in the dual basis $B^*$. Of course, one way to compute a basis whose dual will contain short dual basis is short is to simply compute $B$ such that $B^*$ is a suitably reduced basis of $\lat^*$. Such a basis $B$ is called a dual-reduced basis, and sees use in applications such as~\cite{conf/soda/HavivR14}.

However, we would like to compute a basis such that the vectors in $B$ and $B^*$ are both short, which Seysen addressed in his work~\cite{journals/combinatorica/Seysen93}. Seysen's main result finds a basis $B$ such that both $B$ and $B^*$ are short by dividing this problem into two subproblems. The first involves finding a basis with small $\eta(B)$, as in Section~\ref{sec:slide}. The second subproblem, discussed in~\cite[Section 3]{journals/combinatorica/Seysen93}, involves conditioning unipotent matrices. Let $N(n, \R)$ be the multiplicative group of unipotent $n \times n$-matrices. That is, a matrix $A \in N(n, \R)$ if $a_{ii} = 1$ and $a_{ij} = 0$ for $i > j$ (i.e., $A$ is upper triangular and has ones on the main diagonal). Let $N(n, \Z)$ be the subgroup of $N(n, \R)$ with integer entries. Because $N(n, \Z)$ is a subset of $GL(n, \Z)$, we trivially have that $\lat(B) = \lat(B \cdot U)$ for every $U \in N(n, \Z)$.

Let $\norm{B}_{\infty} := \max_{i,j \in [n]} \abs{b_{ij}}$ denote the largest magnitude of an entry in $B$. We follow Seysen~\cite{journals/combinatorica/Seysen93} and define $S'(B) = \max \set{\norm{B}_{\infty}, \norm{B^{-1}}_{\infty}}$. We also let
\[
  \zeta(n) = \sup_{A \in N(n, \R)} \big\{ \inf_{U \in N(n, \Z)} \set{S'(A \cdot U)} \big\}.
  \]
  
\begin{theorem}[{\cite[Prop. 5 and Thm. 6]{journals/combinatorica/Seysen93}}]
There exists an algorithm \textsc{Seysen} that takes as input $A \in N(n, \R)$ and outputs $A \cdot U$ where $U \in N(n, \Z)$ and $S'(A \cdot U) \leq n^{O(\log n)}$ in time $O(n^3)$. In particular, $\zeta(n) \leq n^{O(\log n)}$.
\label{thm:zeta_bound}
\end{theorem}

Let $B = QR$ be a QR-decomposition of $B$. We may further decompose $R$ as $R = DR'$, where $d_{ii} = \norm{\bt_i}$ and 
\[
r'_{ij} = \left \{
\begin{array}{ll}
  0         & \text{if $j < i$,} \\
  1         & \text{if $j = i$,} \\
  \mu_{ji}   & \text{if $j > i$.} \\
\end{array}
\right.
\]
In particular, note that $R' \in N(n, \R)$. It is easy to see that $\eta(B)$ controls $\| D\| \|D^{-1}\|$. On the other hand, using the bound on $\zeta(n)$, we can always multiply $B$ on the right by $U \in N(n, \Z)$ to control the size of $\|R'\|\| R'^{-1}\|$. Roughly speaking, these two facts imply Theorem~\ref{thm:seysen-main}.

\begin{theorem}[{\cite[Theorem 7]{journals/combinatorica/Seysen93}}]
Let $B = \textsc{Seysen}(B')$ where $B'$ is a matrix. Then $S(B) \leq n \cdot \eta(B') \cdot \zeta(n)^2$.
\label{thm:seysen-main}
\end{theorem}

\begin{proof}[Proof of Theorem~\ref{thm:seysen-slide}]
Let $B = \textsc{Seysen}(B')$, where $B'$ is a basis as computed in Proposition~\ref{prop:slide_alg}. We then have that
\nnote{Was:
  \begin{align*}
    S(B) & \leq n \cdot \eta(B') \cdot \zeta(n)^2 & (\text{by Theorem~\ref{thm:seysen-main}}) \\
    & \leq n \cdot k^{O(n\log k/k)} \cdot \zeta(n)^2 & (\text{by Proposition~\ref{prop:slide_alg}}) \\
    & \leq n \cdot k^{O(n\log k/k)} \cdot (n^{O(\log n)})^2 & (\text{by Theorem~\ref{thm:zeta_bound}}) \\
    & \leq k^{O(n\log k/k)}.
  \end{align*}
}
  \begin{align*}
    S(B) & \leq n \cdot \eta(B') \cdot \zeta(n)^2 & (\text{by Theorem~\ref{thm:seysen-main}}) \\
    & \leq n \cdot k^{O(n/k + \log k)} \cdot \zeta(n)^2 & (\text{by Proposition~\ref{prop:slide_alg}}) \\
    & \leq n \cdot k^{O(n/k + \log k)} \cdot (n^{O(\log n)})^2 & (\text{by Theorem~\ref{thm:zeta_bound}}) \\
    & \leq k^{O(n/k + \log k)}.
  \end{align*}
We can compute $B'$ in $2^{O(k)}$ time using Proposition~\ref{prop:slide_alg}. Moreover, by Theorem~\ref{thm:zeta_bound}, \textsc{Seysen} runs in $O(n^3)$ time. Therefore the algorithm runs in $2^{O(k)}$ time.
\end{proof}

\section{Approximating lattice distortion}
\label{sec:approx}
In this section, we show how to compute low-distortion mappings between lattices by using bases with low $S(B)$.

\subsection{Basis length bounds in terms of \texorpdfstring{$S(B)$}{S(B)}}
Call a basis $B = [\vec{b}_1,\ldots, \vec{b}_n]$ \emph{sorted} if $\norm{\vec{b}_1} \leq \cdots \leq \norm{\vec{b}_n}$. Clearly, $\norm{\vec{b}_i}/\lambda_i \geq 1$ for a sorted basis $B$. Note that sorting $B$ does not change $S(B)$, since $S(\cdot)$ is invariant under permutations of the basis vectors.

A natural way to quantify the ``shortness'' of a lattice basis is to upper bound $\norm{\vec{b}_k}/\lambda_k$ for all $k \in [n]$.
For example,~\cite{journals/combinatorica/LagariasLS90} shows that $\norm{\vec{b}_k}/\lambda_k \leq \sqrt{n}$ when $B$ is an HKZ basis.
We give a characterization of Seysen bases showing that in fact \emph{both} the primal basis vectors and the dual basis vectors are not much longer than the successive minima. Namely, $S(B)$ is an upper bound on both $\norm{\vec{b}_k}/\lambda_k$ and $\norm{\vec{b}_k^*}/\lambda_{n-k+1}^*$ for sorted bases $B$. Although we only use the fact that $S(B) \geq \norm{\vec{b}_k}/\lambda_k$ we show both bounds. Seysen~\cite{journals/combinatorica/Seysen93} gave essentially the same characterization, but we state and prove it here in a slightly different form.

\begin{lemma}[Theorem 8 in~\cite{journals/combinatorica/Seysen93}]
Let $B$ be a sorted basis of $\lat$. Then for all $k \in [n]$,

\begin{enumerate}
\item $\norm{b_k}/\lambda_k(\lat) \leq S(B)$.
\item $\norm{b_k^*}/\lambda_{n-k+1}^*(\lat) \leq S(B)$.
\end{enumerate}
\label{lem:seysen-lb}
\end{lemma}

\begin{proof}
  For every $k \in [n]$ we have
  \begin{align*}
    \norm{\vec{b}_k}/\lambda_k & \leq \norm{\vec{b}_k} \lambda_{n-k+1}^*      & \text{(by the lower bound in Theorem~\ref{thm:transference})} \\
       & \leq \norm{\vec{b}_k} \max_{i \in \set{k, \ldots, n}} \norm{\vec{b}_i^*} & \text{(the $\vec{b}_i^*$ are linearly independent)} \\
       & \leq \max_{i \in \set{k, \ldots, n}} \norm{\vec{b}_i} \norm{\vec{b}_i^*} & \text{($B$ is sorted)} \\
       & \leq S(B).
  \end{align*}
This proves Item 1. Furthermore, for every $k \in [n]$ we have
  \[
\frac{\norm{\vec{b}_k^*}}{\lambda_{n - k + 1}^*} \leq \frac{\norm{\vec{b}_k} \norm{\vec{b}_k^*}}{\lambda_k \lambda_{n - k + 1}^*} \leq \max_{i \in [n]} \norm{\vec{b}_i} \norm{\vec{b}_i^*} = S(B).
  \]  
The first inequality follows from the assumption that $B$ is sorted, and the second follows from the lower bound in Theorem~\ref{thm:transference}. This proves Item 2.
\end{proof}

\subsection{Approximating LDP using Seysen bases}
In this section, we bound the distortion $\dist(\lat_1, \lat_2)$ between lattices $\lat_1, \lat_2$. The upper bound is constructive and depends on $S(B_1), S(B_2)$, which naturally leads to Theorem~\ref{thm:approx-informal}.

\begin{lemma}
Let $A = [\vec{a}_1,\ldots, \vec{a}_n]$ and $B = [\vec{b}_1,\ldots, \vec{b}_n]$ be sorted bases of $\lat_1, \lat_2$ respectively. Then,
\[
\norm{BA^{-1}} \leq n S(A) S(B) M(\lat_1, \lat_2).
\]

\label{lem:dist_ub}
\end{lemma}
\begin{proof}
  \begin{align*}
    \norm{BA^{-1}} &= \Big\|\sum_{i=1}^n \vec{b}_i (\vec{a}_i^*)^T\Big\| \\
    &\leq \sum_{i=1}^n \big\| \vec{b}_i (\vec{a}_i^*)^T\big\| & (\text{by triangle inequality}) \\
    &= \sum_{i=1}^n \norm{\vec{b}_i} \norm{\vec{a}_i^*}  \\
    &\leq n \max_{i \in [n]} \norm{\vec{b}_i} \norm{\vec{a}_i^*} \\
    &\leq n S(B) \max_{i \in [n]} \lambda_i(\lat_2) \norm{\vec{a}_i^*} & (\text{by Item 1 in Lemma~\ref{lem:seysen-lb}}) \\
    &\leq n S(A) S(B) \max_{i \in [n]} \lambda_i(\lat_2)/\norm{\vec{a}_i} & (\text{by definition of $S(A)$}) \\
   & \leq n S(A) S(B) \max_{i \in [n]} \lambda_i(\lat_2)/\lambda_i(\lat_1) & (\text{$A$ is sorted}) \\
   & = n S(A) S(B) M(\lat_1, \lat_2).
  \end{align*}
\end{proof}

\begin{proof}[Proof of Theorem~\ref{thm:basis-distortion-bnd}]
Note that by definition there always exist bases $B_1, B_2$ of $\lat_1, \lat_2$ respectively achieving $S(B_i) = S(\lat_i)$. Therefore, applying Lemma~\ref{lem:dist_ub} twice to bound both $\norm{B_2B_1^{-1}}$ and $\norm{B_1B_2^{-1}}$, we get the upper bound.

For the lower bound, let $\vec{v}_1, \ldots, \vec{v}_n \in \lat_1$ be linearly independent vectors such that $\norm{\vec{v}_i} = \lambda_i(\lat_1)$ for every $i$. Then, for every $i$,
  \begin{align*}
    \lambda_i(\lat_2) \leq \max_{j \in [i]} \norm{T \vec{v}_j} 
                      \leq \norm{T} \max_{j \in [i]} \norm{\vec{v}_j} 
                      = \norm{T} \lambda_i(\lat_1).
  \end{align*}
Rearranging, we get that $\lambda_i(\lat_2)/\lambda_i(\lat_1) \leq \norm{T}$. This holds for arbitrary $i$, so in particular $\max_{i \in [n]} \lambda_i(\lat_2)/\lambda_i(\lat_1) = M(\lat_1, \lat_2) \leq \norm{T}$.
 The same computation with $\lat_1, \lat_2$ reversed shows that $M(\lat_2, \lat_1) \leq \norm{T^{-1}}$. Multiplying these bounds together implies the lower bound in the theorem statement.
\end{proof}

We can now prove Theorem~\ref{thm:approx-informal}.

\begin{proof}[Proof of Theorem~\ref{thm:approx-informal}]
  Let $(\lat_1, \lat_2)$ be an instance of LDP. For $i = 1, 2$, compute a basis $B_i$ of $\lat_i$ using the algorithm described in Theorem~\ref{thm:seysen-slide} with parameter $k$. We have that \nnote{Was: $S(B_i) \leq k^{O(n \log k/k)}$} $S(B_i) \leq k^{O(n/k + \log k)}$. This computation takes $2^{O(k)}$ time. The algorithm then simply outputs $T = B_2B_1^{-1}$. 
  
  By Lemma~\ref{lem:dist_ub} and the upper bounds on $S(B_i)$, we get that \nnote{Was: $\kappa(T) \leq k^{O(n\log k/k)} \cdot M(\lat_1, \lat_2) \cdot M(\lat_2, \lat_1)$} $\kappa(T) \leq k^{O(n/k + \log k)} \cdot M(\lat_1, \lat_2) \cdot M(\lat_2, \lat_1)$. This is within a factor of \nnote{Was: $k^{O(n\log k/k)}$} $k^{O(n/k + \log k)} \cdot n^{O(\log n)} = k^{O(n/k + \log k)}$ of $\dist(\lat_1, \lat_2)$ by \nnote{Was: Theorem~\ref{thm:basis-distortion-bnd}} Theorem~\ref{thm:distortion-bnd}. So, the algorithm is correct.
\end{proof}

\section{Hardness of LDP}
\label{sec:hardness}

In this section, we prove the hardness of $\gamma$-GapLDP. (See Theorem~\ref{thm:LDPhard}.) Our reduction works in two steps. First, we show how to use an oracle for GapLDP to solve a variant of GapCVP that we call $\gamma\text{-}\mathrm{GapCVP}^{\alpha}$. (See Definition~\ref{def:cvpalpha} and Theorem~\ref{thm:CVPalphatoLDP}.) Given a CVP instance consisting of a lattice $\lat$ and a target vector $\vec{t}$, our idea is to compare ``$\lat$ with $\vec{t}$ appended to it'' to ``$\lat$ with an extra orthogonal vector appended to it.'' (See Eq.~\eqref{eq:L1L2}.) We show that, if $\distance(\vec{t},\lat)$ is small, then these lattices will be similar. On the other hand, if (1) $\distance(k\vec{t},\lat)$ is large for all non-zero integers $k$, and (2) $\lambda_1(\lat)$ is not too small; then the two lattices must be quite dissimilar.

We next show that $\gamma\text{-}\mathrm{GapCVP}^{\alpha}$ is as hard as GapSVP. (See Theorem~\ref{thm:CVPalpha}.) This reduction is a variant of the celebrated reduction of~\cite{GMSS99}. It differs from the original in that it ``works in base $p$'' instead of in base two, and it ``adds an extra coordinate to $\vec{t}$.'' We show that this is sufficient to satisfy the promises required by $\gamma\text{-}\mathrm{GapCVP}^{\alpha}$.

Both reductions are relatively straightforward.

\subsection{Reduction from a variant of CVP}

\begin{definition}
\label{def:cvpalpha}
For any $\gamma = \gamma(n) \geq 1$ and $\alpha = \alpha(n) > 0$, $\gamma\text{-}\mathrm{GapCVP}^{\alpha}$ is the promise problem defined as follows. The input is a lattice $\lat \subset \Q^n$, a target $\vec{t} \in \Q^n$, and a distance $d > 0$. It is a `YES' instance if $\distance(\vec{t}, \lat) \leq d$ and a `NO' instance if $\distance(k\vec{t}, \lat) > \gamma d$ for all non-zero integers $k$ \emph{and} $d < \alpha \cdot \lambda_1(\lat)$.
\end{definition}

We will need the following characterization of the operator norm of a matrix in terms of its behavior over a lattice. Intuitively, this says that ``a lattice has a point in every direction.''

\begin{fact}
\label{fact:matrixlatticenorm}
For any matrix $A \in \R^{n\times n}$ and (full-rank) lattice $\lat \subset \R^n$,
\[
\|A\| = \sup_{\vec{y} \in \lat \setminus \{0\}} \frac{\|A\vec{y}\|}{\|\vec{y}\|}
\; .
\]
\end{fact}
\begin{proof}
It suffices to note that, for any $\vec{x} \in \R^n$ with $\|\vec{x} \| = 1$ and any full-rank lattice $\lat \subset \R^n$, there is a sequence $\vec{y}_1,\vec{y}_2,\ldots$ of vectors $\vec{y}_i \in \lat$ such that 
\[
\lim_{m \rightarrow \infty} \frac{\vec{y}_m}{\|\vec{y}_m\|} = \vec{x}
\; .
\]
Indeed, this follows immediately from the fact that the rationals are dense in the reals.
\end{proof}

\begin{theorem}
\label{thm:CVPalphatoLDP}
For any $\gamma = \gamma(n) \geq 1$, there is an efficient reduction from $\gamma'\text{-}\mathrm{GapCVP}^{1/\gamma'}$ to $\gamma$-GapLDP, where $\gamma' = O( \gamma)$.
\end{theorem}
\begin{proof}
On input $\lat \subset \Q^n$ with basis $(\vec{b}_1,\ldots, \vec{b}_n)$, $\vec{t} \in \Q^n$, and $d >0$, the reduction behaves as follows. 
Let $\lat_1 := \lat(\vec{b}_1, \ldots, \vec{b}_n, r \cdot \vec{e}_{n+1})$ with $r > 0$ to be set in the analysis. Let $\lat_2 := \lat(\vec{b}_1, \ldots, \vec{b}_n, \vec{t} + r \cdot \vec{e}_{n+1})$. I.e., 
\begin{equation}
\label{eq:L1L2}
 \lat_1 = \lat\left( \begin{array}{cc} B & \vec{0}\\
 0 & r  \end{array} \right)
 \qquad \qquad 
 \lat_2 = \lat\left( \begin{array}{cc} B & \vec{t}\\
 0 & r  \end{array} \right)
 \; .
\end{equation}
(Formally, we must embed the $\vec{b}_i$ and $\vec{t}$ in $\Q^{n+1}$ under the natural embedding, but we ignore this for simplicity.) The reduction then calls its $\gamma$-GapLDP oracle with input $\lat_1$, $\lat_2$, and $c > 0$ to be set in the analysis and outputs its response.

It is clear that the reduction runs in polynomial time. 
Suppose that $\distance(\vec{t}, \lat) \leq d$. We note that $\lat_2$ does not change if we shift $\vec{t}$ by a lattice vector. So, we may assume without loss of generality that $\vec{0}$ is a closest lattice vector to $\vec{t}$ and therefore $\|\vec{t}\| \leq d$. 

Let $B_1 := [\vec{b}_1, \ldots, \vec{b}_n, r \cdot \vec{e}_{n+1}]$ and $B_2 := [\vec{b}_1, \ldots, \vec{b}_n, \vec{t} + r \cdot \vec{e}_{n+1}]$ be the bases from the reduction. It suffices to show that $\kappa(B_2 B_1^{-1})$ is small. Indeed, for any $\vec{y} \in \lat_1$, we can write $\vec{y} = (\vec{y}',k r)$ for some $k \in \Z$ and $\vec{y}' \in \lat$. Then, we have
\[
\| B_2 B_1^{-1} \vec{y}\| = \|(\vec{y}' + k\vec{t}, kr) \| \leq \|(\vec{y}', kr)\| + |k| \| \vec{t}\| \leq (1+d/r)\| \vec{y}\| 
\; .
\]
Similarly,
$
\| B_2 B_1^{-1} \vec{y}\|  \geq \| \vec{y}  \|  - |k| \|\vec{t}\| \geq (1-d/r) \| \vec{y}\|
$.
Therefore, by Fact~\ref{fact:matrixlatticenorm}, $\kappa(B_2 B_1^{-1}) \leq (1+d/r)/(1-d/r)$. So, we take $c := (1+d/r)/(1-d/r)$, and the oracle will therefore output `YES'.

Now, suppose $\distance(z\vec{t}, \lat) > 10\gamma d$ for all non-zero integers $z$, and $\lambda_1(\lat) > 10 \gamma d$. (I.e., we take $\gamma' = 10\gamma = O(\gamma)$.)
Let $A $ be a linear map with $A \lat_1 = \lat_2$. 
 Note that $A$ has determinant one, so that
$
\kappa(A) \geq \frac{\|A \vec{x}\|}{\|\vec{x} \| }
$
for any $\vec{x} \in \Q^{n+1} \setminus \{ \vec{0} \}$.
We have that  $A (\vec{0}, r) = (\vec{y}', k r )$ for some $\vec{y}' \in \lat + k\vec{t}$ and $k \in \Z$. If $k \neq 0$, then $\|A (\vec{0}, r) \| \geq \distance(k\vec{t}, \lat) > 10\gamma d$. So, $\kappa(A) \geq \|A (\vec{0}, r) \|/r > 10\gamma d/r$.

If, on the other hand, $k = 0$, then $\vec{y}' \in \lat \setminus \{\vec{0} \}$ and
$\|A(\vec{0}, r)\| = \|(\vec{y}', 0)\| \geq \lambda_1(\lat) > 10\gamma d$,
so that we again have
$
\kappa(A) \geq \|A(\vec{0}, r)\|/r > 10\gamma d/r
$.
\nnote{Huck, you thought the previous inequality had a bug in it earlier. Is it ok now?}
Taking $r = 2\gamma d$ gives $\kappa(A) > \gamma \cdot c$, so that the oracle will output `NO', as needed.

\end{proof}

\subsection{Hardness of This Variant of GapCVP}
We recall the definition of (the decision version of) $\gamma$-GapSVP. 
\begin{definition}
For any $\gamma = \gamma(n) \geq 1$, $\gamma\text{-}\mathrm{GapSVP}$ is the promise problem defined as follows: The input is a lattice $\lat \subset \Q^n$, and a distance $d > 0$. It is a `YES' instance if $\lambda_1(\lat) \leq d$ and a `NO' instance if $\lambda_1(\lat) > \gamma d$.
\end{definition}

Haviv and Regev (building on work of Ajtai, Micciancio, and Khot~\cite{Ajtai-SVP-hard,Mic01svp,Khot05svp}) proved the following strong hardness result for $\gamma$-GapSVP~\cite{HRsvp}.

\begin{theorem}[{\cite[Theorem 1.1]{HRsvp}}]
\label{thm:SVPhard}
~ %
\begin{enumerate}
\item $\gamma$-GapSVP is $\NP$-hard under randomized polynomial-time reductions for any constant $\gamma \geq 1$. I.e., there is no randomized polynomial-time algorithm for $\gamma$-GapSVP unless 
$
{\cc{NP} \subseteq \cc{RP}}$.
\item $2^{\log^{1-\eps} n}$-GapSVP is $\NP$-hard under randomized quasipolynomial-time reductions for any constant $\eps > 0$. I.e., there is no randomized polynomial-time algorithm for ${2^{\log^{1-\eps} n}\text{-GapSVP}}$ unless 
$
\cc{NP} \subseteq \cc{RTIME}(2^{\polylog(n)})
$.
\item $n^{c/\log \log n}$-GapSVP is $\NP$-hard under randomized subexponential-time reductions for some universal constant $c > 0$. I.e., there is no randomized polynomial-time algorithm for $n^{c/\log \log n}$-GapSVP unless 
$
\cc{NP} \subseteq \cc{RSUBEXP} := \bigcap_{\delta > 0} \cc{RTIME}(2^{n^{\delta}})
$.
\end{enumerate}
\end{theorem}

In particular, to prove Theorem~\ref{thm:hardness-informal}, it suffices to reduce $\gamma'$-GapSVP to $\gamma\text{-}\mathrm{CVP}^{1/\gamma}$ for $\gamma' = O(\gamma)$.

\begin{theorem}
\label{thm:CVPalpha}
For any $1 \leq \gamma = \gamma(n) \leq \mathrm{poly}(n)$, there is an efficient reduction from $\gamma'$-GapSVP to $\gamma\text{-}\mathrm{GapCVP}^{1/\gamma}$, where $\gamma' =  \gamma \cdot (1+o(1))$.
\end{theorem}
\begin{proof}
Let $p$ be a prime with $10\gamma n \leq p \leq 20 \gamma n \leq \mathrm{poly}(n)$. We take $\gamma' = \gamma \cdot (1+o(1))$ so that \[
\gamma = \frac{\gamma'}{\sqrt{1-\gamma^{\prime 2}/(p-1)^2}}
.
\]

On input a basis $B := [\vec{b}_1,\ldots, \vec{b}_n]$ for a lattice $\lat \subset \Q^n$, and $d > 0$, the reduction behaves as follows. For $i = 1,\ldots, n$, let $\lat_i := \lat(\vec{b}_1,\ldots, p\vec{b}_i, \ldots, \vec{b}_n)$ be ``$\lat$ with its $i$th basis vector multiplied by $p$.'' And, for all $i$ and $1 \leq j < p$, let $\vec{t}_{i,j} := j\vec{b}_i+r\vec{e}_{n+1}$, with $r > 0$ to be set in the analysis. For each $i,j$, the reduction calls its $\gamma\text{-}\mathrm{GapCVP}^{1/\gamma}$ oracle on input $\lat_i$, $\vec{t}_{i,j}$, and $d' := \sqrt{d^2 + r^2}$. 
Finally, it outputs `YES' if the oracle answered `YES' for any query. Otherwise, it outputs `NO'.

It is clear that the algorithm is efficient. Note that
\[
\distance(j \vec{b}_i, \lat_i) = \min \left\{\Big\|\sum_{\ell = 1}^n a_\ell \vec{b}_\ell \Big\| \ : \ a_\ell \in \Z,\  a_i \equiv j \bmod p \right\}
\; .
\]
In particular, $\lambda_1(\lat) = \min_{i,j} \distance(j \vec{b}_i, \lat_i)$.

So, suppose $\lambda_1(\lat) \leq d$. Then, there must be some $i,j$ such that $\distance(\vec{t}_{i,j}, \lat_i)^2 \leq r^2 + \lambda_1(\lat)^2 \leq r^2 + d^2 = d^{\prime 2}$. So, the oracle answers `YES' at least once.

Now, suppose $\lambda_1(\lat) > \gamma' d$. Since $\lat_i \subset \lat$, we have $\lambda_1(\lat_i) \geq \lambda_1(\lat) > \gamma' d$, and therefore $d < \lambda_1(\lat_i)/\gamma' < \lambda_1(\lat_i)/\gamma$, as needed. And, by the above observation, we have $\distance(j \vec{b}_i, \lat_i) \geq \lambda_1(\lat) > \gamma' d$ for all $1 \leq i \leq n$ and $1 \leq j < p$. Furthermore, for any integer $1 \leq z < p$, we have 
$
\distance(z j \vec{b}_i, \lat_i) = \distance((zj \bmod p)\cdot \vec{b}_i, \lat_i) > \gamma' d
$, where we have used the fact that $p$ is prime so that $zj \not\equiv 0 \bmod p$. It follows that
$
\distance(z \vec{t}_{i,j}, \lat_i) > \distance(z j \vec{b}_i, \lat_i) > \gamma' d
$.
And, for $z \geq p$, it is trivially the case that $\distance(z \vec{t}_{i,j}, \lat_i) \geq zr \geq pr$. Taking $r := \gamma' d/(p-1)$, we have that in both cases 
\[
\distance(z \vec{t}_{i,j}, \lat_i) > \gamma' d = \frac{\gamma' d'}{\sqrt{1-r^2}} = \frac{\gamma' d'}{\sqrt{1-\gamma^{\prime 2}/(p-1)^2}} = \gamma d
\; .
\]
So, the oracle will always answer `NO'.
\end{proof}

\begin{corollary}
\label{cor:SVPtoLDP}
For any $1 \leq \gamma = \gamma(n) \leq \poly(n)$, there is an efficient reduction from $\gamma'$-GapSVP to $\gamma$-GapLDP, where $\gamma' = O(\gamma)$.
\end{corollary}
\begin{proof}
Combine Theorems~\ref{thm:CVPalphatoLDP} and~\ref{thm:CVPalpha}.
\end{proof}

With this, the proof of our main hardness result is immediate.

\begin{theorem}
\label{thm:LDPhard}
The three hardness results in Theorem~\ref{thm:SVPhard} hold with GapLDP in place of GapSVP.
\end{theorem}
\begin{proof}
Combine Theorem~\ref{thm:SVPhard} with Corollary~\ref{cor:SVPtoLDP}.
\end{proof}

\section{Some illustrative examples}
\label{sec:examples}
\subsection{Separating distortion from the successive minima}

\label{sec:densepacking}

We now show that, for every $n$, there exists a $\lat$ such that $\dist(\lat,\Z^n) \geq \Omega(\sqrt{n}) \cdot M(\lat, \Z^n) \cdot M(\Z^n,\lat)$. Indeed, it suffices to take any lattice with $\det(\lat)^{1/n} \leq O(n^{-1/2})$ but $\lambda_i(\lat) = \Theta(1)$. (This is true for almost all lattices in a certain precise sense. See, e.g.,~\cite{siegal45}.)

\begin{lemma}
\label{lem:randomlattice}
For any $n \geq 1$,
there is a lattice $\lat \subset \Q^n$ such that $\det(\lat)^{1/n} \leq O(n^{-1/2})$ and $\lambda_i(\lat) = \Theta(1)$ for all $i$.
\end{lemma}

\begin{proposition}
For any $n \geq 1$, there exists a lattice $\lat \subset \Q^n$ such that 
\[
\dist(\lat,\Z^n) \geq \Omega(\sqrt{n}) \cdot M(\lat, \Z^n) \cdot M(\Z^n,\lat)
\; .
\]
\end{proposition}
\begin{proof}
Let $\lat \subset \Q^n$ be any lattice as in Lemma~\ref{lem:randomlattice}. In particular, $M(\lat, \Z^n) \cdot M(\Z^n,\lat) = O(1)$. However, for any linear map $T$ with $T(\lat) = \Z^n$, we of course have 
\[
\|T\| \geq |\det(T)|^{1/n} = \det(\Z^n)^{1/n}/\det(\lat)^{1/n} \geq \Omega(\sqrt{n})
\;.
\] 
(To see the first inequality, it suffices to recall that $|\det(T)| = \prod \sigma_i $ and $\|T\| = \max \sigma_i$, where the $\sigma_i$ are the singular values of $T$.) 
And, $T^{-1}\vec{e}_1$ must be a non-zero lattice vector, 
so $\|T^{-1}\| \geq \|T^{-1}\vec{e}_1\| \geq \lambda_1(\lat) \geq \Omega(1)$. Therefore, $\kappa(T) = \|T\| \|T^{-1}\| \geq \Omega(\sqrt{n})$, as needed.
\end{proof}

\subsection{Non-optimality of HKZ bases for distortion}
\label{sec:hkzlb}
We show an example demonstrating that mappings between lattices built using HKZ bases are non-optimal in terms of their distortion. Let $B_n$ be the $n \times n$ upper-triangular matrix with diagonal entries equal to $1$ and upper triangular off-diagonal entries equal to $-\half$. I.e., $B_n$ has entries 

\[
b_{ij} = \left \{
\begin{array}{cl}
  0                  & \text{if $j < i$,} \\
  1                  & \text{if $j = i$,} \\
  -\half             & \text{if $j > i$.} \\
\end{array}
\right.
\]

Luk and Tracy~\cite{LUK2008441} introduced the family $\set{B_n}$ as an example of bases that are well-reduced but poorly conditioned. Indeed it is not hard to show that $\set{B_n}$ are HKZ bases that nevertheless have $\kappa(B_n) = \Omega(1.5^n)$. We use these bases to show the necessity of using Seysen reduction even on HKZ bases.

\begin{theorem}
For every $n \geq 1$, there exists an $n \times n$ HKZ basis $B$ such that $\dist(\Z^n, \lat(B)) \leq n^{O(\log n)}$, but $\kappa(B) \geq \Omega(1.5^n)$.
\label{thm:hkz-lb}
\end{theorem}

\begin{proof}
  Let $B' = B_n$ be an HKZ basis in the family described above, and take $I_n$ as the basis of $\Z_n$. Then $\kappa(B' \cdot I_n) = \Omega(1.5^n)$.

On the other hand, let $B = \textsc{Seysen}(B')$. Then, because $\eta(B') = 1$, $S(B) = n^{O(\log n)}$ by Theorem~\ref{thm:seysen-main}. Clearly, $\lambda_i(\Z_n) = 1$ for all $i \in [n]$. On the other hand, $1 \leq \lambda_i(\lat(B)) \leq \sqrt{n}$ for all $i \in [n]$. The lower bound holds because $\min \norm{\bt_i} = 1$, and the upper bound comes from the fact that $\norm{\vec{b}'_i} \leq \sqrt{n}$ for all $i \in [n]$ and the linear independence of the $\vec{b}'_i$.\footnote{In fact, $\lambda_n(\lat(B)) = O(1)$.} It follows that $M(\Z^n, \lat(B)) \leq \sqrt{n}$ and $M(\lat(B), \Z^n) \leq 1$. Applying Lemma~\ref{lem:dist_ub} to $B$ and $B^{-1}$, we then get that $\kappa(B \cdot I_n) \leq n^{O(\log n)}$.
\end{proof}

\bibliographystyle{alpha}
\bibliography{distortion}

\begin{thebibliography}{AKKV12}

\bibitem[AFK02]{SFK02}
Sanjeev Arora, Alan Frieze, and Haim Kaplan.
\newblock A new rounding procedure for the assignment problem with applications
  to dense graph arrangement problems.
\newblock {\em Mathematical Programming}, 2002.

\bibitem[Ajt96]{Ajt96}
Mikl{\'o}s Ajtai.
\newblock Generating hard instances of lattice problems.
\newblock In {\em {STOC}}, 1996.

\bibitem[Ajt98]{Ajtai-SVP-hard}
Mikl\'{o}s Ajtai.
\newblock The {S}hortest {V}ector {P}roblem in {L2} is {NP}-hard for randomized
  reductions.
\newblock In {\em STOC}, 1998.

\bibitem[AKKV12]{AKKV12}
Vikraman Arvind, Johannes K{\"o}bler, Sebastian Kuhnert, and Yadu Vasudev.
\newblock {Approximate Graph Isomorphism}.
\newblock In {\em Mathematical Foundations of Computer Science}, 2012.

\bibitem[AKS01]{AKS01}
Mikl\'{o}s Ajtai, Ravi Kumar, and D.~Sivakumar.
\newblock A sieve algorithm for the {S}hortest {L}attice {V}ector {P}roblem.
\newblock In {\em STOC}, pages 601--610, 2001.

\bibitem[Bab16]{BabaiGI16}
L.~Babai.
\newblock {G}raph {I}somorphism in quasipolynomial time, 2016.
\newblock \url{http://arxiv.org/abs/1512.03547}.

\bibitem[Ban93]{bana/transference}
W.~Banaszczyk.
\newblock New bounds in some transference theorems in the geometry of numbers.
\newblock {\em Mathematische Annalen}, 296(1):625--635, 1993.

\bibitem[BV14]{BV14}
Zvika Brakerski and Vinod Vaikuntanathan.
\newblock Lattice-based {FHE} as secure as {PKE}.
\newblock In {\em ITCS}, 2014.

\bibitem[CS98]{ConwaySloaneBook98}
J.~Conway and N.J.A. Sloane.
\newblock {\em Sphere Packings, Lattices and Groups}.
\newblock Springer New York, 1998.

\bibitem[Gen09]{Gen09}
Craig Gentry.
\newblock Fully homomorphic encryption using ideal lattices.
\newblock In {\em {STOC}}, 2009.

\bibitem[GMSS99]{GMSS99}
Oded Goldreich, Daniele Micciancio, Shmuel Safra, and Jean{-}Paul Seifert.
\newblock Approximating shortest lattice vectors is not harder than
  approximating closest lattice vectors.
\newblock {\em Information Processing Letters}, 71(2):55 -- 61, 1999.

\bibitem[GN08]{GN08}
Nicolas Gama and Phong~Q. Nguyen.
\newblock Finding short lattice vectors within {M}ordell's inequality.
\newblock In {\em {STOC}}, 2008.

\bibitem[GPV08]{GPV08}
Craig Gentry, Chris Peikert, and Vinod Vaikuntanathan.
\newblock Trapdoors for hard lattices and new cryptographic constructions.
\newblock In {\em STOC}, 2008.

\bibitem[HR12]{HRsvp}
Ishay Haviv and Oded Regev.
\newblock Tensor-based hardness of the {S}hortest {V}ector {P}roblem to within
  almost polynomial factors.
\newblock {\em Theory of Computing}, 8(23):513--531, 2012.
\newblock Preliminary version in STOC'07.

\bibitem[HR14]{conf/soda/HavivR14}
Ishay Haviv and Oded Regev.
\newblock On the {L}attice {I}somorphism {P}roblem.
\newblock In {\em {SODA}}, 2014.

\bibitem[Kan87]{Kan87}
Ravi Kannan.
\newblock Minkowski's convex body theorem and {I}nteger {P}rogramming.
\newblock {\em Mathematics of Operations Research}, 12(3):pp. 415--440, 1987.

\bibitem[Kho05]{Khot05svp}
Subhash Khot.
\newblock Hardness of approximating the {S}hortest {V}ector {P}roblem in
  lattices.
\newblock {\em Journal of the ACM}, 52(5):789--808, September 2005.
\newblock Preliminary version in FOCS'04.

\bibitem[KZ73]{kzbases}
A.~Korkine and G.~Zolotareff.
\newblock Sur les formes quadratiques.
\newblock {\em Mathematische Annalen}, 6(3):366--389, 1873.

\bibitem[LLL82]{lll/mathannal/lll82}
A.K. Lenstra, H.W. Lenstra{ Jr.}, and L.~Lov\'{a}sz.
\newblock Factoring polynomials with rational coefficients.
\newblock {\em Mathematische Annalen}, 261(4):515--534, 1982.

\bibitem[LLS90]{journals/combinatorica/LagariasLS90}
J.~C. Lagarias, Hendrik~W. Lenstra{ Jr.}, and Claus{-}Peter Schnorr.
\newblock {Korkin-Zolotarev bases and successive minima of a lattice and its
  reciprocal lattice}.
\newblock {\em Combinatorica}, 10(4):333--348, 1990.

\bibitem[LS14]{LenstraSilverberg14}
Hendrik~W. Lenstra{ Jr.} and Alice Silverberg.
\newblock Lattices with symmetry, 2014.
\newblock \url{http://arxiv.org/abs/1501.00178}.

\bibitem[LT08]{LUK2008441}
Franklin~T. Luk and Daniel~M. Tracy.
\newblock An improved lll algorithm.
\newblock {\em Linear Algebra and its Applications}, 428(2):441 -- 452, 2008.

\bibitem[Mic01]{Mic01svp}
Daniele Micciancio.
\newblock The {S}hortest {V}ector {P}roblem is {NP}-hard to approximate to
  within some constant.
\newblock {\em SIAM Journal on Computing}, 30(6):2008--2035, March 2001.
\newblock Preliminary version in FOCS 1998.

\bibitem[Min10]{minkowski1910geometrie}
H.~Minkowski.
\newblock {\em Geometrie der Zahlen}.
\newblock Number v. 1 in Geometrie der Zahlen. B.G. Teubner, 1910.

\bibitem[MR07]{MR07}
Daniele Micciancio and Oded Regev.
\newblock Worst-case to average-case reductions based on {G}aussian measures.
\newblock {\em SIAM Journal on Computing}, 37(1):267--302, 2007.

\bibitem[MV13]{journals/siamcomp/MicciancioV13}
Daniele Micciancio and Panagiotis Voulgaris.
\newblock A deterministic single exponential time algorithm for most lattice
  problems based on {V}oronoi cell computations.
\newblock {\em {SIAM} J. Comput.}, 42(3):1364--1391, 2013.

\bibitem[PS97]{PleskenSouvignier97}
W.~Plesken and B.~Souvignier.
\newblock Computing isometries of lattices.
\newblock {\em J. Symbolic Comput.}, 24(3-4):327--334, 1997.
\newblock Computational algebra and number theory (London, 1993).

\bibitem[Reg09]{Reg09}
Oded Regev.
\newblock On lattices, {L}earning with {E}rrors, random linear codes, and
  cryptography.
\newblock {\em Journal of the ACM}, 56(6):Art. 34, 40, 2009.

\bibitem[Sch87]{Schnorr87}
C.P. Schnorr.
\newblock A hierarchy of polynomial time lattice basis reduction algorithms.
\newblock {\em Theoretical Computer Science}, 1987.

\bibitem[Sey93]{journals/combinatorica/Seysen93}
Martin Seysen.
\newblock Simultaneous reduction of a lattice basis and its reciprocal basis.
\newblock {\em Combinatorica}, 13(3):363--376, 1993.

\bibitem[Sie45]{siegal45}
Carl~Ludwig Siegel.
\newblock A mean value theorem in geometry of numbers.
\newblock {\em Annals of Mathematics}, 46(2):pp. 340--347, 1945.

\bibitem[SSV09]{journals/moc/SikiricSV09}
Mathieu~Dutour Sikiric, Achill Sch{\"{u}}rmann, and Frank Vallentin.
\newblock Complexity and algorithms for computing {V}oronoi cells of lattices.
\newblock {\em Math. Comput.}, 78(267):1713--1731, 2009.

\end{thebibliography}

\end{document}